\documentclass[11pt]{article}

\usepackage[paperwidth=199.8mm,paperheight=297mm,centering,hmargin=25mm,vmargin=25mm]{geometry}
\usepackage{authblk} 
\usepackage[titletoc,title]{appendix}
\usepackage{hyperref}
\hypersetup{colorlinks=true,citecolor=blue,linkcolor=blue,filecolor=blue,urlcolor=blue,breaklinks=true}
\usepackage[dvipsnames]{xcolor}
\usepackage{color}
\usepackage[marginal]{footmisc}
\usepackage{url}

\usepackage{mathrsfs}  
\usepackage{mathtools}
\usepackage{amsmath}
\usepackage{amssymb}
\usepackage[shortlabels]{enumitem}
\usepackage{graphicx,epic,eepic,epsfig,latexsym,verbatim}
\usepackage{dsfont}

\usepackage{framed}
\definecolor{shadecolor}{rgb}{0.9,0.9,0.9}
\usepackage{subcaption}
\usepackage{caption}
\usepackage{float}
\usepackage{tikz}
\usepackage{tcolorbox}
\usepackage{relsize}

\definecolor{LightGray}{RGB}{220,220,220}
\definecolor{myred}{RGB}{236, 17, 0}
\definecolor{myblue}{RGB}{10, 88, 153}
\definecolor{myorange}{RGB}{236, 137, 0}
\definecolor{mygreen}{RGB}{26, 152, 81}

\usepackage{amsthm}

\newtheoremstyle{nonbold}
  {\topsep}
  {\topsep}
  {\upshape}
  {}
  {\itshape}
  {.}
  {.5em}
  {{\thmname{#1}\ \textit{\thmnumber{#2}} \thmnote{\normalfont (#3)}}}

\theoremstyle{nonbold}
\newtheorem{definition}{Definition}

\newtheorem{lemma}[definition]{Lemma}
\newtheorem{theorem}[definition]{Theorem}
\newtheorem{corollary}[definition]{Corollary}

\newcounter{example}

\mathchardef\ordinarycolon\mathcode`\:
\mathcode`\:=\string"8000
\def\vcentcolon{\mathrel{\mathop\ordinarycolon}}
\begingroup \catcode`\:=\active
  \lowercase{\endgroup
  \let :\vcentcolon
  }

\usepackage{colortbl}
\usepackage{pifont}
\definecolor{Gray}{gray}{0.92}
\definecolor{Gray2}{gray}{0.75}
\definecolor{maroon}{cmyk}{0,0.87,0.68,0.32}
\usepackage{booktabs}
\usepackage{makecell}
\usepackage{diagbox}
\usepackage{multirow}

\newcommand{\nc}{\newcommand}
\nc{\rnc}{\renewcommand}
\nc{\bra}[1]{\langle#1|}
\nc{\ket}[1]{|#1\rangle}
\nc{\<}{\langle}
\rnc{\>}{\rangle}
\nc{\ketbra}[2]{|#1\rangle\!\langle#2|}
\nc{\braket}[2]{\langle#1|#2\rangle}
\nc{\braandket}[3]{\langle #1|#2|#3\rangle}
\nc{\proj}[1]{| #1\rangle\!\langle #1 |}
\nc{\avg}[1]{\langle#1\rangle}
\nc{\Rank}{\operatorname{Rank}}
\nc{\smfrac}[2]{\mbox{$\frac{#1}{#2}$}}
\nc{\tr}{\operatorname{Tr}}
\nc{\ox}{\otimes}

\nc{\cA}{{\cal A}}
\nc{\cB}{{\cal B}}
\nc{\cC}{{\cal C}}
\nc{\cD}{{\cal D}}
\nc{\cE}{{\cal E}}
\nc{\cF}{{\cal F}}
\nc{\cG}{{\cal G}}
\nc{\cH}{{\cal H}}
\nc{\cI}{{\cal I}}
\nc{\cJ}{{\cal J}}
\nc{\cK}{{\cal K}}
\nc{\cL}{{\cal L}}
\nc{\cM}{{\cal M}}
\nc{\cN}{{\cal N}}
\nc{\cO}{{\cal O}}
\nc{\cP}{{\cal P}}
\nc{\cQ}{{\cal Q}}
\nc{\cR}{{\cal R}}
\nc{\cS}{{\cal S}}
\nc{\cT}{{\cal T}}
\nc{\cV}{{\cal V}}
\nc{\cU}{{\cal U}}
\nc{\cX}{{\cal X}}
\nc{\cY}{{\cal Y}}
\nc{\cZ}{{\cal Z}}
\nc{\cW}{{\cal W}}

\nc{\RR}{{{\mathbb R}}}
\nc{\CC}{{{\mathbb C}}}
\nc{\FF}{{{\mathbb F}}}
\nc{\NN}{{{\mathbb N}}}
\nc{\ZZ}{{{\mathbb Z}}}
\nc{\QQ}{{{\mathbb Q}}}
\nc{\UU}{{{\mathbb U}}}
\nc{\EE}{{{\mathbb E}}}
\nc{\id}{{\operatorname{id}}}

\nc{\1}{{\mathds{1}}}
\nc{\supp}{{\operatorname{supp}}}

\def\ve{\varepsilon}

\makeatletter
\def\grd@save@target#1{%
  \def\grd@target{#1}}
\def\grd@save@start#1{%
  \def\grd@start{#1}}
\tikzset{
  grid with coordinates/.style={
    to path={%
      \pgfextra{%
        \edef\grd@@target{(\tikztotarget)}%
        \tikz@scan@one@point\grd@save@target\grd@@target\relax
        \edef\grd@@start{(\tikztostart)}%
        \tikz@scan@one@point\grd@save@start\grd@@start\relax
        \draw[minor help lines,magenta] (\tikztostart) grid (\tikztotarget);
        \draw[major help lines] (\tikztostart) grid (\tikztotarget);
        \grd@start
        \pgfmathsetmacro{\grd@xa}{\the\pgf@x/1cm}
        \pgfmathsetmacro{\grd@ya}{\the\pgf@y/1cm}
        \grd@target
        \pgfmathsetmacro{\grd@xb}{\the\pgf@x/1cm}
        \pgfmathsetmacro{\grd@yb}{\the\pgf@y/1cm}
        \pgfmathsetmacro{\grd@xc}{\grd@xa + \pgfkeysvalueof{/tikz/grid with coordinates/major step}}
        \pgfmathsetmacro{\grd@yc}{\grd@ya + \pgfkeysvalueof{/tikz/grid with coordinates/major step}}
        \foreach \x in {\grd@xa,\grd@xc,...,\grd@xb}
        \node[anchor=north] at (\x,\grd@ya) {\pgfmathprintnumber{\x}};
        \foreach \y in {\grd@ya,\grd@yc,...,\grd@yb}
        \node[anchor=east] at (\grd@xa,\y) {\pgfmathprintnumber{\y}};
      }
    }
  },
  minor help lines/.style={
    help lines,
    step=\pgfkeysvalueof{/tikz/grid with coordinates/minor step}
  },
  major help lines/.style={
    help lines,
    line width=\pgfkeysvalueof{/tikz/grid with coordinates/major line width},
    step=\pgfkeysvalueof{/tikz/grid with coordinates/major step}
  },
  grid with coordinates/.cd,
  minor step/.initial=.2,
  major step/.initial=1,
  major line width/.initial=2pt,
}
\makeatother

\makeatletter
\newcommand*\rel@kern[1]{\kern#1\dimexpr\macc@kerna}
\newcommand*\widebar[1]{%
  \begingroup
  \def\mathaccent##1##2{%
    \rel@kern{0.8}%
    \overline{\rel@kern{-0.8}\macc@nucleus\rel@kern{0.2}}%
    \rel@kern{-0.2}%
  }%
  \macc@depth\@ne
  \let\math@bgroup\@empty \let\math@egroup\macc@set@skewchar
  \mathsurround\z@ \frozen@everymath{\mathgroup\macc@group\relax}%
  \macc@set@skewchar\relax
  \let\mathaccentV\macc@nested@a
  \macc@nested@a\relax111{#1}%
  \endgroup
}
\makeatother

\usepackage{cite}

\nc{\sF}{{\mathscr{F}}}
\nc{\sO}{{\mathscr{O}}}
\nc{\bcS}{\boldsymbol{\cS}}
\nc{\bTheta}{\boldsymbol{\Theta}}
\nc{\coneF}{\textsf{Cone}({\sF})}
\nc{\sub}{\text{sub}}
\nc{\density}{\mathscr{D}}
\nc{\qandq}{\quad \text{and} \quad}
\nc{\CPTP}{\text{\rm CPTP}}

\begin{document}

\title{One-shot distillation with constant overhead using catalysts}

\author[1]{Kun Fang \thanks{kunfang@cuhk.edu.cn}}
\author[2]{Zi-Wen Liu \thanks{zwliu0@tsinghua.edu.cn}}

\affil[1]{\small School of Data Science, The Chinese University of Hong Kong, Shenzhen,\protect\\  Guangdong, 518172, China}
\affil[2]{\small Yau Mathematical Sciences Center, Tsinghua University, Beijing 100084, China}

\date{\today}
\maketitle

\begin{abstract}

Quantum resource distillation is a fundamental task in quantum information science and technology. Minimizing the overhead of distillation is crucial for the realization of quantum computation and other technologies. 
Here we explicitly demonstrate how, for general quantum resources, suitably designed quantum catalysts (i.e., auxiliary systems that remain unchanged before and after the process) enable distillation with constant overhead in the practical one-shot setting, thereby overcoming the established logarithmic lower bound for one-shot distillation overhead.
In particular, for magic state distillation, our catalysis method paves a path for tackling the diverging batch size problem associated with code-based low-overhead protocols by enabling arbitrary reduction of the protocol size for any desired accuracy. Notably, this yields constant-overhead magic state distillation with controllable protocol size.
Furthermore, we demonstrate a tunable spacetime trade-off between overhead and success probability enabled by catalysts which offers significant versatility for practical implementation.
Finally, we extend catalysis techniques to dynamical quantum resources and show that channel mutual information determines one-shot catalytic channel transformation, thereby advancing our understanding for both dynamical catalysis and information theory.

\end{abstract}

\section{Introduction}

Quantum technologies hold the potential of offering revolutionary advantages over classical methods in crucial technological tasks including computation and communication, driven by accurate manipulation of various kinds of quantum resources including quantum entanglement~\cite{Horodecki:entanglement}, quantum coherence~\cite{coherenceRMP}, and nonstabilizerness/``magic''~\cite{Bravyi2005}. However, the fragile nature of quantum information resources presents a formidable challenge, as the  microscopic systems are highly susceptible to various sources of interference such as environmental noises and control imperfections~\cite{Preskill2018quantumcomputingin}. These noise and error effects can significantly undermine the effectiveness and security of quantum computation and communication, thereby limiting the scalability and applicability of quantum technologies. In addressing these challenges, quantum resource distillation~\cite{bennett_purification_1996,Bravyi2005}---which entails converting noisy quantum resources into a smaller amount of purer ones with certain free operations---has emerged as a ubiquitously important procedure. For instance, magic state distillation~\cite{Bravyi2005,campbell2010bound,BravyiHaah12,HastingsHaah18,KrishnaTillich18,LBT19,Fang2020PRL,FangLiuPRXQ,Liu2022manybody,wills2024constantoverheadmagicstatedistillation,golowich2024asymptotically,nguyen2024good,lee2024low,golowich2024quantumldpccodestransversal} is of major interest in quantum computation as a leading scheme for the implementation of logical non-Clifford gates, which constitutes the main bottleneck in building fault-tolerant quantum computers.  Other notable examples include entanglement distillation~\cite{BBPS96:ent_dist,bennett_purification_1996,BDSW96:ent_qec,fang2019distillation} and coherence distillation~\cite{WinterYang16,Fang2018,Regula2018,hayashi2021finite} which play key roles in quantum communication, networking and cryptography. 

Given the wide importance of resource distillation procedures, optimizing their efficiency constitutes a vital problem  for quantum information technology. More specifically, the main goal is to reduce the {distillation overhead}, i.e.,~the amount of noisy inputs needed to produce certain desired outputs, for which the key figure of merit is the exponent $\gamma$ in the $O(\log^\gamma (1/\ve))$ asymptotic scaling with respect to target error rate $\ve$ as $\ve\rightarrow 0$.
Mainly for theoretical interest and simplicity, researchers have traditionally focused on the idealized infinite-shot setting where one considers an average notion of overhead over an uncontrollably large amount of outputs. In particular, there have been intensive efforts and advances that gradually reduce the average overhead of magic state distillation since the earliest proposal achieving $\gamma\approx 2.46$~\cite{Bravyi2005}. It was long conjectured and believed that $\gamma$ cannot be reduced to below one~\cite{BravyiHaah12}; however, striking breakthroughs in recent years have broken this barrier, achieving sublogarithmic average overhead~\cite{HastingsHaah18,KrishnaTillich18} and culminating in recent results that eventually attain (almost) constant average overhead~\cite{wills2024constantoverheadmagicstatedistillation,golowich2024asymptotically,nguyen2024good,golowich2024quantumldpccodestransversal}.

Nonetheless, in realistic scenarios, both the number of available inputs and the desired outputs are always finite. This corresponds to the one-shot regime which is fundamentally distinct from the above infinite-shot setting and directly captures resource-processing tasks in practice. In this practical setting, the recent universal no-go theorems for resource purification (which applies to virtually any type of resource including  entanglement, coherence, asymmetry, thermal nonequilibrium, and magic)~\cite{Fang2020PRL,FangLiuPRXQ} rigorously imply the existence of a $\gamma\geq 1$ barrier, i.e., a logarithmic lower bound on distillation overhead. Notably, although there exist protocols achieving sublogarithmic average overhead in the asymptotic setting (as mentioned above), they inevitably suffer from a common ``batch size problem'' from a practical perspective: they usually involve a prohibitively large number of states or dimensions  (e.g., the Hastings--Haah sublogarithmic protocol~\cite{HastingsHaah18} works with a code that entails $\geq2^{58}$ input and $\geq 2^{44}$ output states), and the protocol volume must diverge rapidly as $\ve\rightarrow 0$, which incurs enormous practical costs.

In this work, we introduce a distillation scheme utilizing the idea of quantum catalysis, a phenomenon analogous to chemical catalysis where catalysts can facilitate reactions by providing alternative pathways with lower activation energies, demonstrating how it helps circumvent the above limitations. In the realm of quantum information, the catalysis phenomenon involves an auxiliary quantum system, namely the catalyst, which participates in the transformation process, enabling transformations that would otherwise be unachievable while remaining unaltered before and after the process and thus can be reused.
Quantum catalysis was first discussed in the context of entanglement transformation~\cite{jonathan1999entanglement} more than two decades ago. 
{For example, by the majorization criterion~\cite{Nielsen1999,nielsen2002introduction}, an entangled state $\ket{\psi_1}=\sqrt{0.4}\ket{00}+\sqrt{0.4}\ket{11}+\sqrt{0.1}\ket{22}+\sqrt{0.1}\ket{33}$ cannot be transformed into $\ket{\psi_2}=\sqrt{0.5}\ket{00}+\sqrt{0.25}\ket{11}+\sqrt{0.25}\ket{22}$ via local quantum operations and classical communication. However, there exists a catalyst state $\ket{\phi}=\sqrt{0.6}\ket{44}+\sqrt{0.4}\ket{55}$ that enables the exact transformation from $\ket{\psi_1}\ox\ket{\phi}$ to $\ket{\psi_2}\ox\ket{\phi}$. Beyond entanglement manipulation, quantum catalysis}
has since been extensively studied across various resource theories, including quantum coherence~\cite{aaberg2014catalytic,Takagi2022PRL}, thermodynamics~\cite{muller2018correlating,Shiraishi2021,Wilming2021PRL},  purity~\cite{boes2019neumann} and magic~\cite{campbell2011catalysis,gidney2018halving,beverland2020lower}. This area remains active as evidenced by recent works~\cite{char2023catalytic,datta2023there,Kondra2021PRL,Lipka-Bartosik2021,Shiraishi2021,Bavaresco2025}, with comprehensive reviews available in Refs.~\cite{datta2023catalysis,lipka2024catalysis}. While catalytic methods are widely useful, recent studies have also revealed fundamental limitations of their power, e.g., they~{cannot} overcome bound entanglement~\cite{lami2023catalysis}, and {cannot} increase distillable entanglement~\cite{ganardi2023catalytic}. These ``negative'' results highlight that whether catalysis can introduce advantages for a certain task is a nontrivial question~\cite{karvonen2021neither}. Furthermore, existing literature on distillation mostly focuses on the distillation rate, namely how much pure target resource can be extracted from a fixed amount of noisy input resource. This is conceptually dual to the overhead considered in this work, where the aim is to minimize the amount of input for some given target.

The main contribution of this work is to demonstrate the effectiveness of catalytic methods in enhancing distillation efficiency, particularly highlighting their striking capability to arbitrarily reduce one-shot distillation overhead, which is unknown before. This, in some sense, completes the picture of distillation overhead optimization.
Specifically, we establish a general and rigorous ordering of resource overheads for one-shot (unassisted), multi-shot average, and one-shot catalytic distillation settings, as illustrated in Fig.~\ref{fig:surpassing}, by showing how any multi-shot distillation protocol can be converted into a one-shot catalytic protocol while maintaining the overhead  by designing suitable catalysts with reusability guarantees. 
\begin{figure}[]
    \centering
    \includegraphics[width=0.6\textwidth]{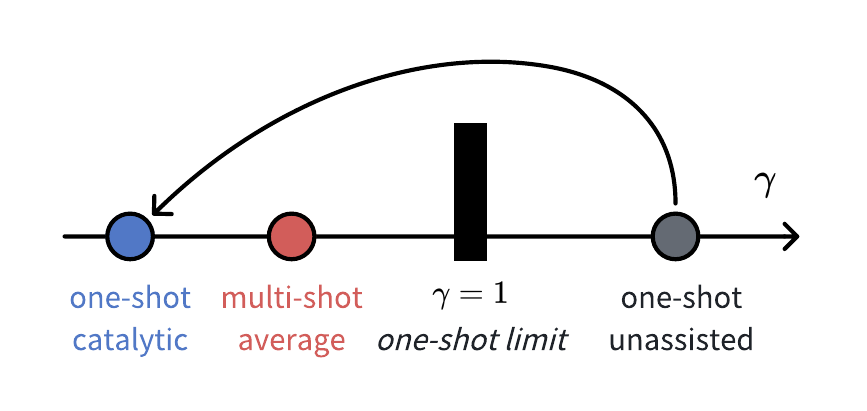}
    \vspace{-0.4cm}
    \caption{Comparison of distillation overhead $\log^\gamma(1/\varepsilon)$ across different settings. A fundamental no-go limit of $\gamma \geq 1$ has been shown to exist for one-shot (unassisted) distillation~\cite{Fang2020PRL,FangLiuPRXQ}. We show in this work that this limit can be surpassed with the aid of quantum catalysts.}
    \label{fig:surpassing}
\end{figure}
In particular, our catalytic approach can be applied to make the batch size of any protocol arbitrarily small (e.g., an $n$-to-$m$ protocol can be converted into a $\lceil n/m \rceil$-to-$1$ protocol) for any accuracy. {Through magic state distillation, this provides a strategy for alleviating the qubit cost in fault-tolerant quantum computation.} 
Furthermore, we demonstrate that the failure probability and overhead, corresponding to time and space costs respectively, can be traded for each other with the help of catalysts. Notably, we prove that the optimal constant for constant-overhead magic state distillation can even be reduced to $1$ at the price of sacrificing the success probability by a constant factor, pushing the efficiency of distillation to its ultimate limit.
Lastly, we extend the catalytic technique to the channel (dynamical) setting. As an application, we provide a one-shot operational interpretation of channel mutual information, showing that it determines the catalytic convertibility of quantum
channels and thereby resolving the channel version of a conjecture of Wilming~\cite{Wilming2021PRL}.
In Fig.~\ref{fig:comparison}, we illustrate the fundamental advantages of our catalytic methods more clearly, including a comparison with prominent existing distillation protocols.
Furthermore, from a practical perspective, our catalytic schemes offer substantial flexibility in experiment and architecture design by enabling one to freely tailor the qubit and time costs to suit specific hardware features or capabilities.  It is worth noting that our protocols consist of quite simple operations and are thus expected to be particularly friendly for experiments.

{
\begin{figure}[htb]
    \centering
    \includegraphics[width=\textwidth]{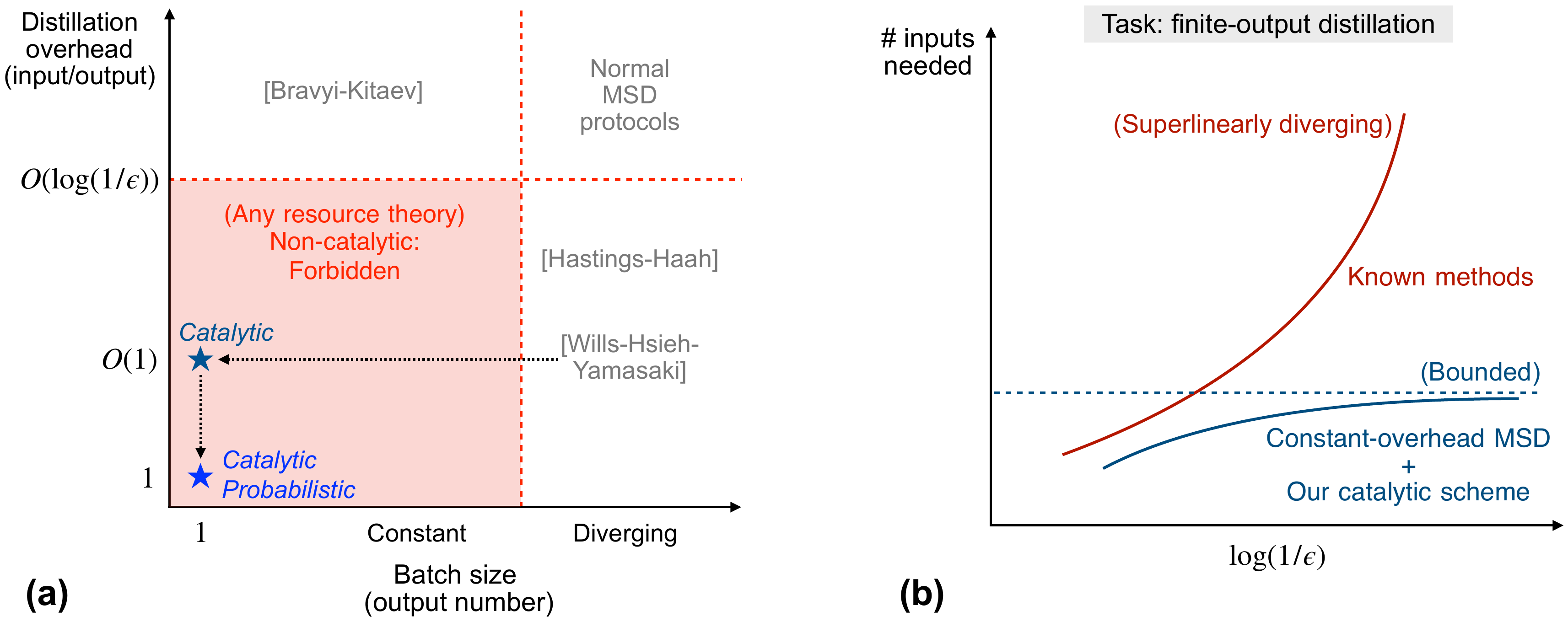}
\caption{{Illustration of the fundamental advantages of our catalytic scheme over existing distillation methods. (a) A distillation overhead versus batch size diagram. Due to the no-go theorems in Refs.~\cite{Fang2020PRL,FangLiuPRXQ}, for virtually any well-behaved resource theory, at least one of the following must hold for non-catalytic protocols (where $\epsilon$ denotes the target error rate): i) the distillation overhead is at least  logarithmic, i.e., $\Omega(\log(1/\epsilon))$; ii) the distillation cannot be one-shot, that is, the batch size (number of outputs) must diverge as $\epsilon$ decreases (which we refer to as the ``batch size problem''). Consequently, the lower-left region is fundamentally forbidden (highlighted in red). For the particularly important case of magic state distillation, normal protocols achieve neither (i) nor (ii) and therefore fall in the upper-right region; the logarithmic-overhead  barrier~\cite{BravyiHaah12} was only recently overcome by e.g.,~Hastings--Haah~\cite{HastingsHaah18} and Wills--Hsieh--Yamasaki~\cite{wills2024constantoverheadmagicstatedistillation}, which are regarded as major breakthroughs, yet they necessarily suffer from the batch size problem. Our catalytic strategy enables one to bypass these barriers and simultaneously achieve (i) and (ii): for example, our catalytic strategy can make the batch size of the constant-overhead protocol~\cite{wills2024constantoverheadmagicstatedistillation} arbitrarily small while preserving the constant overhead (Sec.~\ref{sec:msd}), and by allowing probabilistic success we can further achieve arbitrarily small overhead (Sec.~\ref{sec:prob}). (b) For finite-output distillation tasks,  existing protocols methods require a superlinearly diverging amount of input resources with respect to $\log(1/\epsilon)$ decreases, whereas our catalytic magic state distillation schemes can achieve arbitrarily small $\epsilon$ with finite input resources.}}
    \label{fig:comparison}
\end{figure}
}

\section{Results}

\subsection{Preliminaries: general resource theory and distillation overhead}

To present our results in a rigorous and unified manner, we shall adopt the language of general resource theory, a framework that finds great success in studying different quantum resource features in recent years (background and setup can be found in e.g.~Refs.~\cite{ChitambarGour19,LBT19,Fang2020PRL,PhysRevX.9.031053}). A standard resource theory is defined by a set of {free operations} $\sO$ and a set of {free states} $\sF \subseteq \density$ (the complement of which are resource states) where $\density$ denotes the set of all density matrices.  Taking the magic theory as an example, $\sF$ consists of stabilizer states, and  Clifford operations are a standard choice of $\sO$.    Generally speaking, $\sF$ and $\sO$ can be adaptively defined, leading to a wide variety of meaningful resource theories, as long as they follow a {golden rule}: any free operation can only map a free state to another free state, i.e.~$\cE(\rho)\in\sF, \forall\rho\in\sF, \forall\cE\in\sO$. 

The general goal of distillation tasks is to transform noisy resources into pure ones by some protocol represented by a free operation. In practice, we are also interested in protocols that only produce desired outputs with a certain probability (e.g.~protocols based on error correction, which work upon passing the syndrome measurements).      
To encompass such probabilistic cases, consider the generalization of {free operations} $\sO$ to the class {of free subnormalized quantum operations, i.e.,~completely positive and trace-nonincreasing maps} $\sO_{\sub}:=\{\cL \,|\, \forall \rho \in \sF, \exists\, p \geq 0, \psi \in \sF, \text{ s.t. } \cL(\rho) = p \cdot \psi \}$. {A free probabilistic protocol that transforms $\rho$ to $\gamma$ with probability $p$ is modelled by a quantum operation $\cE_{A\to XB}(\rho_A) = \ket{0}\bra{0}_X \ox \cL_{A\to B}(\rho_A) + \ket{1}\bra{1}_X \ox \cG_{A\to B}(\rho_A)$, {where $\cL, \cG \in \sO_{\sub}$}. Here $X$ is an external flag register that keeps track of whether the protocol succeeds (0) or not (1); then $\cL \in \sO_{\sub}$ represents the successful transformation such that $\cL_{A\to B}(\rho_A) = p \gamma_B$ where $p=\tr [\cL(\rho)]$.}
The case where $\cL$ is a completely positive trace-preserving (CPTP) map and thus $p=1$ corresponds to a deterministic protocol.

As already emphasized in the introduction, previous research on distillation overhead has focused on the asymptotic average overhead in the infinite-shot limit for simplicity, whereas the one-shot setting with controllable output size is practically more relevant and constitutes the main target in this work. 

Now we lay down formal definitions of the different types of overhead.
Denote the space for the source and target states as $S$, and the ancillary system that supports the catalyst as $A$.
For any two quantum states $\rho,\sigma \in \density(S)$, define their trace distance by $\Delta(\rho, \sigma):= \|\rho-\sigma\|_1/2$ with $\|\cdot\|_1$ being the trace norm.
Denote $\rho^{\ox n}\xrightarrow[]{(\ve,p)} \sigma^{\ox m}$ if there exists a resource transformation protocol $\cL\in \sO_{\sub}$ such that $\cL(\rho^{\ox n}) = p \cdot \eta^m$ with $\eta^m \in \density(S^m)$ and $\Delta(\eta_i^m, \sigma) \leq \ve$
for all $ i \in \{1,\cdots,m\}$, where {$\eta^m$ is some (not necessarily product) state in the $m$-copy space and} $\eta_i^m$ is the $i$-th marginal of $\eta^m$. {Here we follow the convention in  distillation studies and focus on the error of each marginal states~\cite{Bravyi2005,BravyiHaah12}. 
Since pure targets are typically of interest, this is equivalent to defining the error over the global states}. 
\begin{definition}[Unassisted  overhead]
    Let $\rho,\sigma \in \density(S)$ be two arbitrary quantum states. Considering $\rho$ and $\sigma$ as the primitive (source) state and target state of distillation respectively, the (conventional unassisted notions of)  {one-shot distillation overhead} $C_{\ve,p}$ and {multi-shot average distillation overhead} $\widebar{C}_{\ve,p}$ are defined as 
    \begin{align}
        C_{\ve,p}(\rho,\sigma)  &\coloneqq \min \{n: \rho^{\ox n}\xrightarrow[]{(\ve,p)} \sigma\},\\
        \widebar{C}_{\ve,p}(\rho,\sigma)  &\coloneqq \min \{\lceil n/m \rceil: \rho^{\ox n}\xrightarrow[]{(\ve,p)} \sigma^{\ox m}\},
    \end{align}
    respectively.
\end{definition}

Furthermore, denote $\rho^{\ox n} \ox \omega \xrightarrow[]{(\ve,p)}\sigma \ox \omega$ if there exists a transformation protocol $\cL \in \sO_{\sub}$ and a catalyst $\omega \in \density(A)$ such that $\cL(\rho^{\ox n}\ox \omega) = p \cdot \nu$ with $\nu \in \density(SA)$, $\Delta(\nu_S, \sigma_S) \leq \ve$ and $\nu_A = \omega_A$, 
where $\nu_A,\nu_S$ are two marginals of $\nu$. 

\begin{definition}[One-shot catalytic distillation overhead]
Let $\rho,\sigma \in \density(S)$ be two arbitrary quantum states. Considering $\rho$ and $\sigma$ as the primitive (source) state and target state of distillation respectively, 
the {one-shot catalytic distillation overhead} $\widetilde{C}_{\ve,p}$ is defined as 
\begin{align}
    \widetilde C_{\ve,p}(\rho, \sigma) := \min \left\{n: \rho^{\otimes n} \ox \omega \xrightarrow[]{(\ve,p)} \sigma \ox \omega, \exists\,\omega \in \density\right\}.
\end{align}
\end{definition}

The standard distillation setting of usual interest is to distill copies of less noisy versions of some target pure state by consuming more copies of more noisy versions of this target state, namely, $\sigma$ is pure and $n \geq m$. For example, the typical task in magic state distillation is to distill less but better $T$ states from noisy ones.

Previous results on distillation overhead are mainly bounds on the multi-shot average notion obtained by taking the asymptotic limit $n,m\rightarrow\infty$ which corresponds to the idealistic infinite-resource scenario. In contrast, the one-shot setting captures the practical finite-resource scenario.
With the above notations, the no-go theorem of Ref.~\cite{Fang2020PRL} sets the logarithmic fundamental limit $C_{\ve,p}(\rho, \sigma) = \Omega(\log (1/\ve))$ for ordinary unassisted one-shot distillation overhead. Our aim here is to demonstrate how the catalytic version $\widetilde C_{\ve,p}$ can break this barrier while maintaining the one-shot nature.

\subsection{General theory: one-shot catalytic resource distillation}

Now we introduce our main results. All proofs and technical analysis are deferred to the Methods section.  
Here, we deal with the standard distillation setting where more source states are consumed to distill less pure target states, namely, $\sigma$ is pure and $n \geq m$.
Axiomatically, it is assumed that appending free systems, discarding and permuting subsystems, and classically controlled free operations are free, and no other properties of the resource theory are needed, rendering the framework highly flexible and applicable to most relevant theories.

Since a one-shot distillation protocol can be regarded as a multi-shot distillation protocol with further conditions, it follows directly that $\widebar{C}_{\ve,p}(\rho, \sigma) \leq C_{\ve,p}(\rho, \sigma)$ for any quantum states $\rho$ and $\sigma$. Here we adapt a catalysis technique originally proposed in Ref.~\cite{duan2005multiple} for LOCC (local operations with classical communication) transformations of quantum entanglement (further explored in recent works such as~\cite{char2023catalytic,datta2023there,Kondra2021PRL,Lipka-Bartosik2021,Shiraishi2021,Takagi2022PRL,Wilming2021PRL,Bavaresco2025})  to our general transformation task, which involves an unequal number of source and target states. This is achieved by regrouping the source states into blocks and embedding the target state into a larger space. 
Particularly, we show that any multi-shot distillation protocol from $\rho^{\ox n}$ to $\sigma^{\ox m}$ can effectively be turned into a one-shot catalytic protocol from $\rho^{\lceil n/m\rceil} \ox \omega$ to $\sigma \ox \omega$ with the same parameters. Here, the catalyst $\omega$ plays a crucial role by enabling new transformation pathways that are otherwise unattainable and absorbing part of the complexity of the transformation process. Although the catalyst temporarily participates in the transformation, it is fully restored afterward.  As a result, we have the inequality $\widetilde{C}_{\ve,p}(\rho, \sigma) \leq \widebar{C}_{\ve,p}(\rho, \sigma)$. In conclusion, the general relation of distillation overhead across different settings is as follows: 
\begin{theorem}
\label{thm: overheads}
For any quantum states $\rho$ and $\sigma$, target error $\ve\in [0,1]$ and success probability $p\in[0,1]$, the following relation holds: $\widetilde{C}_{\ve,p}(\rho, \sigma) \leq \widebar{C}_{\ve,p}(\rho, \sigma) \leq C_{\ve,p}(\rho,\sigma)$.
In particular, any multi-shot distillation protocol can be turned into an one-shot catalytic distillation protocol with the same overhead, implying the first inequality.
\end{theorem}

It is worth noting that the catalyst technique has been studied in the context of distillation before. However, previous works focused on optimizing the output resource that can be extracted from some fixed noisy input, for which known results are negative. For instance, it has been shown that catalysts {cannot} overcome bound entanglement~\cite{lami2023catalysis} and {cannot} increase distillable entanglement~\cite{ganardi2023catalytic}. In contrast, this work is intended to initiate a different perspective by establishing the advantages of catalysts in enhancing distillation overhead, where the aim is to minimize the amount of noisy input required to achieve a desired target. Theorem~\ref{thm: overheads} establishes a foundational result in this new direction, demonstrating the potential of catalysts to reduce the distillation overhead for the first time. In the sections that follow, we confirm the effectiveness of catalysts in this regard through concrete examples, particularly in the context of magic state distillation, highlighting their capability to surpass previously established limits.

A key advantage of using a catalyst is its recoverability after the transformation, allowing for repeated reuse. We now provide a theoretical guarantee that the catalyst can be reused indefinitely without any degradation in its expected performance.

\begin{theorem}\label{thm: catalyst reuse}
For a deterministic one-shot catalytic resource distillation procedure, after $l \geq 1$ repeated uses of the catalyst $\omega_A$, we obtain a joint state $\nu_{S_1S_2\cdots S_l A}$ such that the catalyst is exactly returned on its marginal $\nu_A = \omega_A$ and the target states $\nu_{S_1} = \nu_{S_2} = \cdots = \nu_{S_l}$ with error $\Delta(\nu_{S_i}, \sigma_{S}) \leq \ve$ for all $i \in \{1,\cdots,l\}$.    
\end{theorem}

\vspace{0.2cm}

An apparent subtlety is that the catalyst may exhibit correlations with the remaining systems. However, when addressing distillation tasks aimed at pure target states which is typically the case of interest, the correlation between the obtained target state and the remaining systems is under control. {This can be made explicit by applying Lemma 8 from the Supplemental Material of Ref.~\cite{ganardi2023catalytic}, which states that an error threshold on the target state $\Delta(\nu_S, \ket{\phi}\bra{\phi}_S) \leq \ve$ ensures that the correlation between the target state and the catalyst remains small, with $\Delta(\nu_{SA}, \ket{\phi}\bra{\phi}_S \ox \nu_A) \leq \ve + 3 \sqrt{\ve}$. Consequently, any intermediate usage of the target state and the catalyst will have a negligible impact on the other systems. In the following, we give a further improved bound for this fact through a simplified proof.}

\begin{lemma}\label{lem: pure trace distance}
    Suppose $\Delta(\nu_S,\ket{\phi}\bra{\phi}_S) \leq \ve$, then $\Delta(\nu_{SA},\ket{\phi}\bra{\phi}_S \ox \nu_A) \leq 2\sqrt{\ve}$.
\end{lemma}

We also note that the catalyst is guaranteed to remain effective in deterministic protocols, while for probabilistic ones there is a risk of losing the validity of the catalyst when the transformation fails which arises from the probabilistic nature of multi-shot protocols. Nevertheless, it is guaranteed that the probability of losing the catalyst in the converted one-shot catalytic protocol is no greater than~$p$ if the multi-shot protocol fails with probability $p$. Moreover, the failure probability is finite and can be significantly suppressed by e.g.~concatenation, and there are ways to make probabilistic distillation protocols deterministic~\cite{heussen2025magicstatedistillationmeasurements}. {Finally, we stress that the magic state distillation protocols  we apply later are already deterministic~\cite{wills2024constantoverheadmagicstatedistillation}, so the aforementioned risk is not present.}

\subsection{Applications in magic state distillation}
\label{sec:msd}

As a potential resolution to the non-Clifford gate bottleneck of fault-tolerant quantum computing, magic state distillation has been a subject of extensive study since Bravyi and Kitaev's proposal in 2005~\cite{Bravyi2005}. 
The primary focus is to reduce the resource overhead of magic state distillation, with numerous improvements achieved over the years~\cite{Bravyi2005,BravyiHaah12,Haah2018codesprotocols,meier2012magic,campbell2012magic,jones2013multilevel,HastingsHaah18,KrishnaTillich18,wills2024constantoverheadmagicstatedistillation,golowich2024asymptotically,nguyen2024good,lee2024low,golowich2024quantumldpccodestransversal}.

It was conjectured~\cite{BravyiHaah12} and long believed that the overhead exponent $\gamma$ cannot be reduced to below one. 
In the one-shot setting, this can indeed be proven by applying the universal no-go theorem for resource purification in Ref.~\cite{Fang2020PRL}.
Remarkably, Hastings and Haah first found a code that achieves an average  overhead of $n/m = O(\log^\gamma(1/\ve))$ for $\gamma \approx 0.678$~\cite{HastingsHaah18}, and very recent progress further shows that the exponent can be made arbitrarily close to $0$ or even to equal $0$, namely achieving  constant average overhead~\cite{wills2024constantoverheadmagicstatedistillation,golowich2024asymptotically,nguyen2024good}.  As explained, despite these advances, the low overhead in these protocols comes with the serious problem that they require uncontrollably large batch sizes as $\ve$ tends to zero. 
By applying the catalytic conversion established in Theorem~\ref{thm: overheads}, we can turn the low-overhead distillation protocols into one-shot counterparts that achieves arbitrary accuracy with an arbitrary amount of (even only one) target states. This provides a way to circumvent the fundamental limit of $\Omega(\log(1/\ve))$ for one-shot distillation~\cite{Fang2020PRL}.

\begin{corollary}
There exist one-shot catalytic magic state distillation protocols that achieve any given target error with {certainty (i.e., success probability $1$)} and constant overhead. {In particular,  by Theorem~\ref{thm: catalyst reuse}, the catalyst can be reused indefinitely without degradation in its performance.}
\end{corollary}

As discussed, the correlation between the catalyst and the obtained target state can be made sufficiently small.
A related result by Rubboli and Tomamichel~\cite{rubboli2022fundamental} indicates that making such residual correlations arbitrarily small requires a divergent amount of resources in the catalyst if the resource theory has multiplicative maximum fidelity of resource (i.e., $\hat{F}(\rho_1 \otimes \rho_2) = \hat{F}(\rho_1) \cdot \hat{F}(\rho_2)$, where $\hat{F}(\rho) := \max_{\sigma \in \sF} F(\rho,\sigma)$ with $F$ being the normal state fidelity). However, the maximum fidelity in magic theory (also named stabilizer fidelity) is known to be non-multiplicative~\cite[Section 6.2]{Bravyi2019simulationofquantum}. It remains an interesting open question whether in magic theory it is possible to find exotic catalysts with a finite amount of magic that achieve vanishing residual correlations or error.

\subsection{Catalytic spacetime conversion: trading success probability for reduced overhead}
\label{sec:prob}

In distillation tasks, the resource overhead essentially represents the space requirement, characterizing the size of quantum systems that need to be coherently manipulated in a single experiment. On the other hand, the success probability is associated with  the number of experimental repetitions needed to accomplish the tasks and thus determines the time cost. As the size of the qubits and coherent control is limited, especially in the near term, it is usually more feasible to scale up an experiment by repeating the experiment multiple times rather than conducting it with larger systems in fewer trials. By leveraging quantum catalysis, we enable a trade-off between space and time, allowing  the resource burden to be shifted from qubit overhead to time.

More specifically, we show that if a transformation from $\rho^{\otimes n}$ to $\sigma^{\otimes m}$ is achievable with success probability $p$, then a catalytic transformation from $\rho^{\otimes k}$ to  $\sigma$ can also be achieved with success probability $p m / \lceil n/k \rceil$, for any $1 \leq k \leq n/m$, effectively trading time for space. Notably, by setting $k = 1$, we obtain a one-shot catalytic distillation protocol with unit overhead and success probability $p m / n$, in which case the overhead is minimized to its extreme.

\begin{theorem}\label{thm: reduce the overhead}
Suppose there exists a distillation protocol transforming $\rho^{\ox n}$ to $\sigma^{\ox m}$ with success probability $p$ and target error $\ve$. Then it holds for any $1\leq k \leq n/m$ that $\widetilde{C}_{\ve,pm \left\lceil n/k \right\rceil^{-1}}(\rho, \sigma) \leq k$. In particular, $\widetilde{C}_{\ve,pm/n}(\rho, \sigma) = 1$.
\end{theorem}


It should be noted that trading success probability for overhead is not always possible without the use of a catalyst, which underscores the significance of Theorem~\ref{thm: reduce the overhead}. Consider an example task of entanglement distillation where the goal is to transform  a noisy entangled state $\rho$ (e.g., an isotropic state) into a Bell state $\psi$ with success probability $p'$ and target error $\ve'$. There exist nontrivial forbidden tuples $(p', \ve')$  according to the no-go theorems established in Refs.~\cite{Fang2020PRL,PhysRevLett.128.110505}, which are nevertheless achievable in the average case using the hashing protocol~\cite{devetak2005distillation}: $\rho^{\otimes mr}$ can be transformed into $\psi^{\otimes m}$ with arbitrarily small target error and success probability $1$ for certain $r$ and $m$, where $1/r$ corresponds to the hashing bound. This shows that achieving a multi-copy transformation does not necessarily guarantee the feasibility of a one-shot transformation by simply compromising the success probability, even when $p'$ is chosen to be close to zero. Yet, this trade-off can always be achieved with the aid of a catalyst by our result.

As mentioned, very recent works \cite{wills2024constantoverheadmagicstatedistillation,golowich2024asymptotically,nguyen2024good} demonstrated that magic state distillation can be achieved with constant overhead in the asymptotic limit but left the question of optimality of this constant open for future investigation. {In particular, the protocol in Ref.~\cite{wills2024constantoverheadmagicstatedistillation} achieves $p=1$ and $n/m = O(1)$.} By applying Theorem~\ref{thm: reduce the overhead}, we find that the optimal constant can be ultimately reduced to $1$ with the aid of a catalyst, at the expense of compromising the success probability by a constant factor.

\begin{corollary}
There exist one-shot catalytic magic state distillation protocols that achieve any given target error with constant success probability using only one copy of the source magic state {(i.e., constant overhead $1$)}. {In particular,  by Theorem~\ref{thm: catalyst reuse}, the catalyst can be reused indefinitely without degradation in its performance.}
\end{corollary}

\subsection{Channel catalytic transformation theory}

The concept of catalytic transformation can also be extended to the manipulation of quantum channels, which correspond to dynamical resources that play a crucial role in e.g.~quantum communication and quantum error correction \cite{Bennett_2014,gour2019comparison,LiuWinter2018,LiuYuan:channel,fang2018quantum,fang2019quantum,gour2020dynamicalresources,FangLiuPRXQ,RegulaTakagi21}. 
In this section, we extend the catalysis technique to the channel setting. 

Let $L(A\to B)$ be the set of all linear maps from Hilbert space $A$ to Hilbert space $B$. Let $\CPTP(A\to B)$ be the set of all quantum channels (i.e., completely positive and trace-preserving maps) from $A$ to $B$. The most general transformation of a quantum channel is described by a superchannel, which is a linear map that maps quantum channels to quantum channels and has a physical realization with pre- and post-processing on the quantum channel upon which it acts. More specifically, $\Pi$ is a superchannel~\cite{gour2019comparison} if there exists a Hilbert space $E$ with $|E| \leq |A_0||B_0|$ and two CPTP maps $\Gamma_{\text{pre}} \in \CPTP(B_0\to A_0E)$ and $\Gamma_{\text{post}} \in \CPTP(A_1E\to B_1)$ such that for all linear map $\Psi \in L(A_0\to A_1)$, it holds that $\Pi[\Psi] = \Gamma_{\text{post}} \circ(\Psi \ox \cI_E) \circ \Gamma_{\text{pre}}$.

Before addressing catalytic channel manipulation, we first need to clarify the definition of a reduced channel. It is known that if a bipartite channel satisfies the non-signaling (NS) condition $\tr_{B_2} \circ \cN_{A_1A_2 \to B_1 B_2} = \tr_{B_2} \circ \cN_{A_1A_2 \to B_1 B_2}\circ \cR^\pi_{A_2}$ where $\cR^\pi$ is a replacer channel that sends any input to a fixed state $\pi$, then we can define a unique channel such that $\tr_{B_2} \circ \cN_{A_1A_2\to B_1B_2} = \cN_{A_1\to B_1} \circ \tr_{A_2}$. This channel is given by
\begin{align}\label{eq: reduced channel 1} \cN_{A_1\to B_1}(\cdot) := \tr_{B_2} \cN_{A_1A_2\to B_1B_2}((\cdot) \ox \pi_{A_2}). 
\end{align}
By the NS assumption, it can be easily verified that the definition in~\eqref{eq: reduced channel 1} is independent of the choice of $\pi$. Therefore, without loss of generality, we can always take $\pi$ to be the maximally mixed state of appropriate dimension. See also discussions in Ref.~\cite[Section 2.2]{berta2022semidefinite}.
More generally, if the channel $\cN_{A_1A_2\to B_1B_2}$ does not satisfy the NS condition, we can still define a linear map via~\eqref{eq: reduced channel 1}, which is clearly a quantum channel from $A_1 \to B_1$. The only difference is that the channel will depend on the choice of $\pi$.

 Define the error between quantum channels $\cN$ and $\cM$ to be $\Delta(\cN,\cM):= \frac{1}{2}\|\cN-\cM\|_{\diamond}$  where $\|\cE\|_\diamond:= \sup_{k\in \NN} \sup_{\|X\|_1\leq 1} \|\cE\ox \cI_k(X)\|_1$ denotes the diamond norm~\cite{kitaev1997quantum}. We formally characterize catalytic channel transformations as follows.

\begin{definition}
Let $\cN \in \CPTP(A\to B)$ and $\cM \in \CPTP(A'\to B')$ be two quantum channels. We say $\cN$ can be catalytically transformed to $\cM$ with target error $\ve$ if there exists a free superchannel $\Pi$ and a channel $\cC \in \CPTP(X \to Y)$ such that the output channel $\cP := \Pi_n [\cN \ox \cC_n] \in \CPTP(A'X \to B'Y)$ satisfies
\begin{align}
\cP_{X\to Y} = \cC_{X \to Y},\quad \text{and}\quad
\Delta(\cP_{A'X\to B'Y}, \cM_{A'\to B'} \ox \cC_{X\to Y}) \leq \ve,\label{eq: catalytic channel trans 1}
\end{align}
where $\cP_{X\to Y}(\cdot) := \tr_{B'} \cP_{A'X\to B'Y}((\cdot) \ox \pi_{A'})$.
Moreover, we say  $\cN$ can be (asymptotically) catalytically transformed to $\cM$  if there exists a sequence of superchannels $\Pi_n$ and a sequence of catalyst channel $\cC_n$ such that $\cN^{\ox n}$ can be catalytically transformed to $\cM^{\ox n}$ with error $\ve_n$ and $\lim_{n\to \infty} \ve_n = 0$. 
\end{definition}

The first condition in~\eqref{eq: catalytic channel trans 1} implies that after the coding strategy $\Pi$, the transmission scheme $\cP_{X\to Y}(\cdot)$ functions the same as the channel $\cC$. That is, the functioning of $\cC$ is not affected by the coding, thus it serves as a catalyst channel. The second condition in~\eqref{eq: catalytic channel trans 1} implies that the correlation between the source channel and the catalyst channel is within a target error. This particularly implies that the reduced channel $\cP_{A'\to B'} (\cdot) := \tr_Y \cP_{A'X\to B'Y}((\cdot) \ox \pi_X)$
will function approximately as $\cM$, in the sense that $\Delta(\cP_{A'\to B'}, \cM_{A'\to B}) \leq  \ve$ (see Lemma~\ref{lem: marginal channel monotone} in Methods). 

With the above definition, we are now ready to discuss quantum channel simulation. By the well-known quantum reverse Shannon theorem~\cite{Bennett_2014}, we know that channel manipulation under entanglement-assisted or NS-assisted codes is reversible. This, in particular, implies that $\cN$ can be asymptotically transformed to $\cM$ with unit rate if and only if $I(\cN) \geq I(\cM)$, where $I(\cN):= \sup_{\varphi_{RA}} I(B:R)_{\cN_{A\to B}(\varphi_{RA})}$ 
is the quantum channel mutual information and $I(B:R)_\rho = S(\rho_B) + S(\rho_R) - S(\rho_{BR})$ with $S(\rho)=-\tr \rho \log \rho$,
giving an operational meaning of channel mutual information in the conventional i.i.d. setting. 

Extending the catalysis technique for quantum states, we obtain the following result which shows that $\cN$ can be catalytically transformed to one $\cM$ if and only if $I(\cN) \geq I(\cM)$ given free entanglement and catalyst channels, thereby addressing the channel version of the open question raised in Ref.~\cite{Wilming2021PRL} regarding whether mutual information determines the convertibility of quantum channels.  

\begin{theorem}\label{thm: app mutual information}
Let $\cN,\cM \in \CPTP(A\to B)$ be two quantum channels. Then $\cN$ can be catalytically transformed to $\cM$ with entanglement-assisted or NS-assisted codes if and only if $I(\cN) \geq I(\cM)$.
\end{theorem}

This endows the channel mutual information with an operational meaning in terms of catalytic transformation at the one-shot level, enriching our understanding of quantum catalysis from a fundamental information-theoretic perspective. 

\section{Discussion}

In this work, we considered quantum resource distillation with the assistance of catalysis and presented universal protocols with rigorous efficiency and reusability guarantees that surpass known fundamental limitations of one-shot distillation overhead. Our results highlight the utility and power of catalytic methods in enabling flexible configuration of different types of resources in quantum information processing tasks, which has strong practical appeal.
Specifically, in the context of magic state distillation, our methods provide an experimentally viable strategy for optimizing the cost of fault-tolerant quantum computing.

There remain various subtle issues that may arise in practical implementations, including noise and correlation effects on the catalysts, which warrant further investigation and optimization.
An important step forward is to study the benefits and challenges in implementing our catalytic distillation protocols with regard to specific experimental setups. These protocols are anticipated to be highly experiment-friendly due to the simplicity of required operations, specifically,  classically controlled operations and relabeling or possibly collective movements that are e.g.~natively efficient with the reconfigurable atom array platform, which emerged as a promising candidate  for experimental fault-tolerant quantum computing recently~\cite{Bluvstein_2023} with  unique advantages in overcoming the crucial non-Clifford gate barrier~\cite{Wang_2024}. 

{
Regarding the fundamental theory of catalytic distillation, an important avenue for further study is to relax the requirement of exact catalyst preservation and instead allow the catalyst to be returned only approximately, in the spirit of entanglement embezzling~\cite{van2003universal,leung2013coherent,LeungWang2014,zanoni2024complete}. We expect that this more flexible setting would allow further improvements of specific parameters of catalytic protocols, including the overhead constant and the size of catalyst for a given target error.} 
In addition, while the catalysts used in this work generally have a size comparable to the system in multi-shot protocols, necessitating coherent manipulations of large quantum systems, the prospect of designing smaller, more efficient catalysts remains an open and valuable direction. 
{Conversely, the fundamental bounds on the catalyst size under certain target parameters are also compelling problems to further investigate (see Sec.~\ref{sec:msd} for relevant discussion).}

Moreover, given the connection between magic state distillation and quantum error-correcting (QEC) codes, our results may also provide new perspectives on catalytic QEC codes~\cite{brun2006correcting,brun2014catalytic}. Finally, extending these results to continuous variable systems and conducting a systematic study of channel theory remain promising directions for future work.

\section{Methods}

In this section, we provide proofs and technical analysis of the main results.

\subsection{Proof of Theorem~\ref{thm: overheads}}

The second inequality follows from their definitions. We now prove the first inequality by constructing an explicit catalytic transformation with the desired performance. As mentioned, the catalyst structure we use here was originally proposed in Ref.~\cite{duan2005multiple} in the context of LOCC transformations of quantum entanglement. As mentioned, while similar catalyst structures have been used in various contexts recently~\cite{char2023catalytic,datta2023there,Kondra2021PRL,Lipka-Bartosik2021,Shiraishi2021,Takagi2022PRL,Wilming2021PRL,Bavaresco2025}, existing works mostly consider the transformation of source states to an equal number of target states. 
{Here we adapt the strategy to the distillation of general resource types with arbitrary numbers of input and output states, to show that any multi-shot distillation protocol that transforms $\rho^{\ox n}$ to $\sigma^{\ox m}$ can be turned into an one-shot catalytic protocol that transforms $\rho^{\lceil n/m\rceil} \ox \omega$ to $\sigma \ox \omega$ with the same parameters (which then implies the inequalities in the theorem statement).} 
{In the proof, we present the procedure in a general and rigorous manner, and provide an illustrative example of its steps in Fig.~\ref{fig:catalyst_state}.} Consider any protocol $\cL\in \sO_{\sub}$ such that $\cL(\rho^{\ox n}) = p \eta^m$ and $\Delta(\eta_i^m, \sigma) \leq \ve$
for all $ i \in \{1,\cdots,m\}$, where $\eta^m \in \density(S^m)$ and $\eta_i^m$ is the $i$-th marginal of $\eta^m$. {For our purposes, we prepare $m$ blocks of quantum states, each block containing $k = \lceil n/m \rceil$ copies of $\rho$. We denote a full block by $\zeta = \rho^{\ox k}$.} 
Since $mk = m \lceil n/m \rceil \geq n$, we have a free operation $\cL_1$ such that 
\begin{align}
    \cL_1(\zeta^{\ox m}) = \cL_1(\rho^{\ox mk}) = p 
    \eta^m \qandq \Delta(\eta_i^m, \sigma) \leq \ve,
\end{align}
for all $ i \in \{1,\cdots,m\}$.
This can be done by simply discarding the residual $mk-n$ copies of the source states and then performing the transformation $\cL$ on $\rho^{\ox n}$.
Note that $\eta_m$ and $\sigma^{\ox m}$ live in $\density(S^m)$, but we can embed them into a larger space $\density((S^k)^m)$ so that the individual subsystem matches the space of $\zeta$. More precisely, consider embedding $\cE(\cdot) = (\cdot) \ox \pi_{S}^{\ox (k-1)}$ where $\pi\in\sF$ is a free state, and let $\hat \eta^m = \cE^{\ox m}(\eta^m)$ and $\hat \sigma = \cE(\sigma)$. Consider $\cL_2 = \cE \circ \cL_1$ we get a transformation
\begin{align}
\cL_2(\zeta^{\ox m}) = p \hat \eta^m \qandq \Delta(\hat \eta^m_i, \hat \sigma) \leq \ve,
\end{align}
where $\hat \eta^m_i$ is the marginal state on the $i$-th system and the inequality follows since trace distance is invariant under our embedding. 

Now, we construct a catalyst $\omega$ to transform $\zeta$ to $\hat \sigma$ within the given error. Let $\hat \eta^m_{1:i}$ be the marginal state of $\hat \eta_m$ on the first $i$ systems. Consider the catalyst
\begin{align}\label{eq: catalyst}
    \omega := \frac{1}{m}\sum_{i=1}^{m} \zeta^{\ox i-1} \ox \hat \eta^m_{1:m-i} \ox \ket{i}\bra{i},
\end{align}
where $\ket{i}\bra{i} \in \density(F)$ are states in a classical flag register $F$.
This gives an overall state 
\begin{align}
    \zeta \ox \omega = \frac{1}{m}\sum_{i=1}^m \zeta^{\ox i} \ox \hat \eta^m_{1:m-i} \ox \ket{i}\bra{i}.
\end{align}
We now aim to use this catalyst to complete the desired transformation. {Let $\cP_0(\cdot):= P_0 (\cdot) P_0$ with $P_0:= \sum_{i=1}^{m-1} \ket{i}\bra{i}$ and $\cP_1(\cdot):= P_1 (\cdot) P_1$ with $P_1:= \ket{m}\bra{m}$ be two projecting maps on the flag register $F$.
First, performing a classically controlled transformation $\cI \ox \cP_0 + \frac1p\cL_2 \ox \cP_1$ on $\zeta \ox \omega$, we get
\begin{align}
\nu_1 & := \left(\cI \ox \cP_0 + \frac1p\cL_2 \ox \cP_1\right) (\zeta \ox \omega)\\
& = \frac{1}{m}\sum_{i=1}^{m-1} \zeta^{\ox i} \ox \hat{\eta}^m_{1:m-i} \ox \ket{i}\bra{i} + \frac{1}{m} \hat \eta^m \ox \ket{m}\bra{m},
\end{align}
where $\cI$ is an identity map. This succeeds with probability $p$. Note that this transformation is simply an one-way LOCC operation~\cite{Chitambar_2014}, where we measure the classical flag register and feed forward the classical outcome to adaptively apply the quantum operation $\cI$ or $\cL_2$. Therefore, this classically controlled operation is clearly free.} Second, by cyclically permuting the classical registers of $\nu_1$ such that $i\to i+1$ and $m \to 1$, we get
\begin{align}\label{eq: nu 2}
    \nu_2 := \frac{1}{m} \sum_{i=1}^m \zeta^{\ox i-1} \ox \hat \eta^m_{1:m-i+1} \ox \ket{i}\bra{i}.
\end{align}
Third, permuting the quantum registers in a similar way such that $i\to i+1$ and $m \to 1$, we get the resulting state $\nu_3$. Note that $\nu_3$ and $\nu_2$ are the same state but with different system orders. Finally, note that the embedding operation is reversible by removing the ancillary registers. Denote $\cE^{-1}(\cdot) = \tr_{S^{k-1}}(\cdot)$ to be the reverse embedding operation and $\nu := \cE^{-1} \ox \cI^{\ox m-1}(\nu_3)$ be the final state after reverse embedding.

After these transformations, we can show that (i) the marginal state on the first system of $\nu$ is our target state satisfying the target error; (ii) the marginal state on the last $m-1$ systems of $\nu$ is exactly the catalyst $\omega$. For the first claim, the marginal state on the first system of $\nu$ (or equivalently, the marginal state on the last system of $\nu_2$ after reverse embedding) is given by
\begin{align}
\nu' := \cE^{-1}\left(\frac{1}{m} \sum_{i=1}^m \hat \eta^m_i\right)  = \frac{1}{m} \sum_{i=1}^m \eta^m_i.   
\end{align}
Since $\Delta(\eta^m_i, \sigma) \leq \ve$ we have $\Delta(\nu', \sigma) \leq \ve$ by the convexity of trace distance. The second claim is evident by checking that the first $m-1$ systems of $\nu_2$ gives $\omega$ in Eq.~\eqref{eq: catalyst}. This gives a one-shot catalytic distillation protocol with the same performance in target error $\ve$ and success probability $p$, thereby completing the proof.

\begin{figure}[ht]
    \centering
    \includegraphics[width=\textwidth]{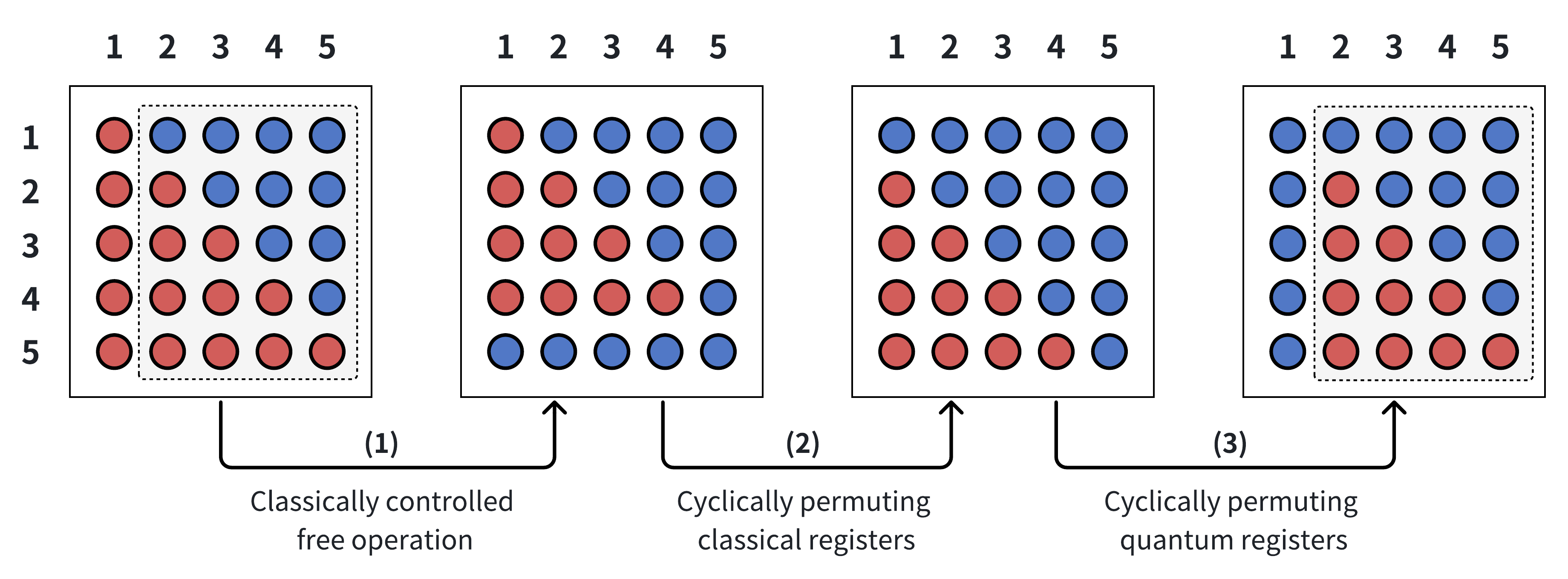}
    \caption{An illustration for the steps of our catalytic procedure used to prove Theorem~\ref{thm: overheads}. {In this example, $n=15$, $m=5$, and hence $k = 3$.} 
     Each dot represents a quantum state on the system $S^3$. Red dots correspond to groups of states $\zeta = \rho^{\otimes 3} \in \density(S^3)$, while blue dots represent embedded states $\hat{\sigma} = \sigma \otimes \pi^{\otimes 2} \in \density(S^3)$ that match the system of $\zeta$. 
     The overall quantum state is a mixture of rows, each labeled with a classical register.  {Suppose we have a $15$-to-$5$ distillation protocol.} In the first step, a classically controlled free operation is applied, {using this protocol to transform} the last row from $\zeta^{\otimes 5} = \rho^{\otimes 15}$ to $\hat{\sigma}^{\otimes 5}$. The second step involves cyclically permuting the classical registers (i.e., the rows of dots), and the third step involves cyclically permuting the quantum registers (i.e., the columns of dots). {Overall, this gives a 3-to-1 catalytic distillation protocol from $\zeta = \rho^{\ox 3}$ to $\hat \sigma$ (which is equivalent to $\sigma$ by tracing out the free states $\pi^{\ox 2}$).} The dashed boxes highlight the catalyst state, showing that it remains unchanged before and after the transformation.}
    \label{fig:catalyst_state}
\end{figure}

\subsection{Proof of Theorem~\ref{thm: catalyst reuse}}

Let $T = S^k$ be the systems on the source state and $\zeta_T = \rho_S^{\ox k}$. Let $\cL$ be the catalytic transformation. In the first round, we get $\cL_{TA\to S_1A}(\zeta_{T} \ox \omega_A) = \nu_{S_1A}$ with $\Delta(\nu_{S_1},\sigma_{S_1}) \leq \ve$, and $\nu_A = \omega_A$.
Then we apply $\cL$ again to a fresh copy of the source state $\zeta_T$ together with the catalyst and leave the state on the system $S_1$ untouched. Denote the output systems as $S_2A$ where $S_2$ is isomorphic to $S_1$. This gives the global state $\cL_{TA\to S_2A}(\zeta_T\ox \nu_{S_1A})$.
Then we can check the catalyst on the system $A$ by
$\tr_{S_2S_1}\cL_{TA\to S_2A}(\zeta_T\ox \nu_{S_1A}) =  \tr_{S_2}\cL_{TA\to S_2A}(\zeta_T\ox \omega_{A}) = \omega_A$.
So the marginal state on the catalytic system is unchanged. We can also check the target state on the system $S_2$ by $\tr_{S_1A} \cL_{TA\to S_2A}(\zeta_T\ox \nu_{S_1A}) = \tr_{A} \cL_{TA\to S_2A}(\zeta_T\ox \omega_A) = \tr_{A} \cL_{TA\to S_1A}(\zeta_T\ox \omega_A) = \nu_{S_1}$, which means that the marginal state we distill remains exactly the same as in the first round. The stated result is obtained by applying the above process repeatedly $l$ times.

\subsection{Proof of Lemma~\ref{lem: pure trace distance}}

 Let $F(\rho,\sigma) = \tr \sqrt{\rho^{1/2} \sigma \rho^{1/2}}$ be the fidelity. Since $\Delta(\nu_S,\ket{\phi}\bra{\phi}_S) \leq \ve$, we have the relation that $F(\nu_S, \ket{\phi}\bra{\phi}_S) \geq \sqrt{1-\ve}$. As $\nu_{SA}$ is an extension of $\nu_S$, by Uhlmann's theorem~\cite{uhlmann1976transition}, there exists an extension of $\ket{\phi}\bra{\phi}_S$ such that its fidelity with $\nu_{SA}$ is the same as $F(\nu_S, \ket{\phi}\bra{\phi}_S)$. However, since $\ket{\phi}\bra{\phi}_S$ is a pure state, its extension can only be in the form of a tensor product. Therefore, we can assume its extension as $\ket{\phi}\bra{\phi}_S \ox \mu_A$ such that $F(\nu_{SA}, \ket{\phi}\bra{\phi}_S \ox \mu_A) = F(\nu_S, \ket{\phi}\bra{\phi}_S)$. By the data processing inequality of fidelity, we get $F(\nu_A,\mu_A) \geq F(\nu_{SA}, \ket{\phi}\bra{\phi}_S \ox \mu_A) \geq \sqrt{1-\ve}$. This implies that $\Delta(\nu_{SA}, \ket{\phi}\bra{\phi}_S \ox \mu_A) \leq \sqrt{\ve}$ and $\Delta(\nu_A,\mu_A) \leq \sqrt{\ve}$. By the triangle inequality of trace distance, we have $\Delta(\nu_{SA}, \ket{\phi}\bra{\phi}_S \ox \nu_A) \leq \Delta(\nu_{SA}, \ket{\phi}\bra{\phi}_S \ox \mu_A) + \Delta(\ket{\phi}\bra{\phi}_S \ox \mu_A, \ket{\phi}\bra{\phi}_S \ox \nu_A) \leq 2\sqrt{\ve}$.

\subsection{Proof of Theorem~\ref{thm: reduce the overhead}}

Let $\cL\in \sO_{\sub}$ be a resource distillation protocol that transforms $n$ copies of the source state $\rho_S$ to $m$ ($m\leq n$) copies of the target state $\sigma_S$ with success probability $p$ and within target error $\ve$. That is, there exists a quantum state $\eta^m \in \density(S^m)$, such that $\cL(\rho^{\ox n}) = p \eta^m$ and $\Delta(\eta_i^m, \sigma) \leq \ve$
for all $ i \in \{1,\cdots,m\}$, where $\eta_i^m$ is the $i$-th marginal of $\eta^m$.
Then we can consider $g = \left\lceil n/k \right\rceil \geq m$ groups of $\zeta = \rho^{\ox k}$. Since $gk = k\left\lceil n/k \right\rceil \geq n$, there exists a free operation $\cL_1$ such that $\cL_1(\zeta^{\ox g}) = \cL_1(\rho^{\ox gk}) = p \eta^m$ and $\Delta(\eta_i^m, \sigma) \leq \ve$, for any $i \in \{1,\cdots,m\}$. This can be done by simply abandoning the residual $gk-n$ copies of the source states and then performing the transformation $\cL$ on $\rho^{\ox n}$. Then we need to prove that there exists a catalytic transformation $\cL'$ that transforms one copy of $\zeta$ to one copy of $\sigma$ with success probability $p m/g$ and within target error $\ve$. An illustration for the  steps of the procedure is given in Fig.~\ref{fig:spacetime tradeoff}. 

Before the actual transformation, we need to perform two embedding steps to ensure that the systems match. First, as $\zeta \in \density(S^k)$ and $\sigma \in \density(S)$, we need to embed $\sigma$ into a larger space by using the embedding $\cE(\cdot) = (\cdot) \otimes \pi_S^{\ox (k-1)}$ where $\pi \in \sF$ is a free state. Let $\hat \eta^m = \cE^{\ox m}(\eta^m)$ and $\hat \sigma = \cE(\sigma)$. Then, since $m \leq g$, we need to do a second embedding to compensate this system mismatching by using free states. Moreover, to avoid mixing the free state with our target state which could compromise the final fidelity, we need to append the free states on an orthogonal Hilbert space so that we can distinguish them via projective measurements later. More explicitly, let $\theta$ be a free state on the system $S^k$. We embed it into a joint Hilbert space $T = S^k W$ by $\tilde \theta_T = \theta_{S^k} \ox \ket{1}\bra{1}_W$. We then embed $\zeta, \hat{\sigma} \in \density(S^k)$ into the larger space $T$ by $\tilde{\zeta}_T = \zeta_{S^k} \ox \ket{0}\bra{0}_W$ and $\tilde{\sigma}_T = \hat \sigma_{S^k} \ox \ket{0}\bra{0}_W$, respectively. Similarly, we embed $\hat{\eta}^m$ into $\density(T^m)$ by $\tilde{\eta}^m = \hat{\eta}^m \ox \ket{0}\bra{0}_{W^{m}}$.  

Then there exists a distillation protocol $\cL_2$ that transforms $\tilde \zeta^{\ox g} \in \density(T^g)$ to $\tilde \eta_m \ox \tilde \theta^{\ox g-m} \in \density(T^g)$. Let $\tilde \eta^m_{1:i}$ be the marginal state of $\tilde \eta_m$ on the first $i$ systems and $\tilde \eta^m_i$ be the marginal state on the $i$-th system. 
Define the following catalyst state on $\density(T^{g-1}F)$ with classical register $F$:
\begin{align}
\omega := \frac1g \sum_{i=1}^{g-m-1} \tilde \zeta^{\ox i-1} \ox \tilde \eta^{m} \ox \tilde \theta^{\ox g - m - i} \ox \ket{i}\bra{i} + \frac1g \sum_{i=g-m}^{g} \tilde\zeta^{\ox i-1} \ox \tilde\eta^m_{1:g-i} \ox \ket{i}\bra{i}.
\end{align}
Note that for $m=g$ or $m=g-1$ the first term vanishes.
We have the overall state on $\density(T^gF)$,
\begin{align}
    \tilde \zeta \ox \omega = \frac1g \sum_{i=1}^{g-m-1} \tilde \zeta^{\ox i} \ox \tilde \eta^{m} \ox \tilde\theta^{\ox g- m - i} \ox \ket{i}\bra{i} + \frac1g \sum_{i=g-m}^{g} \tilde\zeta^{\ox i} \ox \tilde \eta^m_{1:g-i} \ox \ket{i}\bra{i} .
\end{align}
{Let $\cP_0(\cdot):= P_0 (\cdot) P_0$ with $P_0:= \sum_{i=1}^{g-1} \ket{i}\bra{i}$ and $\cP_1(\cdot):= P_1 (\cdot) P_1$ with $P_1:= \ket{g}\bra{g}$ be two projecting maps on the flag register $F$.
We then proceed with the following steps. First, performing a classically controlled transformation $\cI \ox \cP_0 + \frac1p\cL_2 \ox \cP_1$ on $\tilde \zeta \ox \omega$, we obtain
\begin{align}
\nu_1 := & \left(\cI \ox \cP_0 + \frac{1}{p}\cL_2 \ox \cP_1\right) (\tilde \zeta \ox \omega)\\
= & \frac1g \sum_{i=1}^{g-m-1} \tilde \zeta^{\ox i} \ox \tilde \eta^m \ox \tilde \theta^{\ox g - m - i} \ox \ket{i}\bra{i}\\
& + \frac1g \sum_{i=g-m}^{g-1} \tilde \zeta^{\ox i} \ox \tilde \eta^m_{1:g-i} \ox \ket{i}\bra{i} + \frac1g \tilde \eta^{m} \ox \tilde \theta^{\ox g-m}\ox \ket{g}\bra{g}.
\end{align}
As explained in the proof of Theorem~\ref{thm: overheads}, this procedure is a free operation.}
Second, by cyclically permuting the classical registers of $\nu_1$  such that $i\to i+1$ and $g\to 1$, we get 
\begin{align}
\nu_2 := 
& \frac1g \sum_{i=1}^{g-m} \tilde \zeta^{\ox i-1} \ox \tilde \eta^{m} \ox \tilde \theta^{\ox g - m - i + 1} \ox \ket{i}\bra{i}
 + \frac1g \sum_{i=g-m+1}^{g} \tilde \zeta^{\ox i-1} \ox \tilde \eta^m_{1:g-i+1} \ox \ket{i}\bra{i}.
\end{align}
Third, permuting the quantum registers in a similar way such that  $i\to i+1$ and $g\to 1$, we get the resulting state $\nu_3$. Note that $\nu_3$ and $\nu_2$ are the same state but with different system orders.
Then we can check that the last $g-1$ quantum systems of $\nu_3$ returns the catalyst, or equivalently, we can check the first $g-1$ systems of $\nu_2$ and get
\begin{align}
\frac1g \sum_{i=1}^{g-m} \tilde \zeta^{\ox i-1} \ox \tilde \eta^{m} \ox \tilde \theta^{\ox g - m - i} \ox \ket{i}\bra{i}
 + \frac1g \sum_{i=g-m+1}^{g} \tilde \zeta^{\ox i-1} \ox \tilde \eta^m_{1:g-i} \ox \ket{i}\bra{i}
    =   \omega.
\end{align}
The first quantum system of $\nu_3$ or equivalently the last quantum system of $\nu_2$ gives our target state:
\begin{align}
    \nu' = \frac1g \sum_{i=1}^{g-m} \tilde \theta + \frac1g \sum_{i=g-m+1}^g \tilde \eta^m_{g-i+1} = \frac{m}{g} \nu'' + \left(1-\frac{m}{g}\right) \tilde \theta, \quad \text{with} \quad \nu'' = \frac{1}{m}\sum_{i=1}^m \tilde \eta^m_i.
\end{align}
Since $\nu'' = \frac{1}{m}\sum_{i=1}^m \tilde \eta^m_i = \frac{1}{m}\sum_{i=1}^m \hat \eta^m_i \ox \ket{0}\bra{0}_W$ and $\tilde \theta = \theta_{S^k} \ox \ket{1}\bra{1}_W$ are orthogonal, we can perform a projective measurement on the flag system $W$ to distinguish these two states. So we will have a probability $m/g$ to obtain the state $\nu''$. Finally, as all appended states are uncorrelated to the target system , we get $\nu''' = \cE^{-1}(\tr_W \nu'') = \frac{1}{m}\sum_{i=1}^m \eta^m_i$ by removing all the appended systems for free. Since $\Delta(\eta^m_i, \sigma) \leq \ve$, we have $\Delta(\nu''', \sigma) \leq \ve$ by the convexity of trace distance. 
This concludes the proof.

\begin{figure}[ht]
    \centering
    \includegraphics[width=\textwidth]{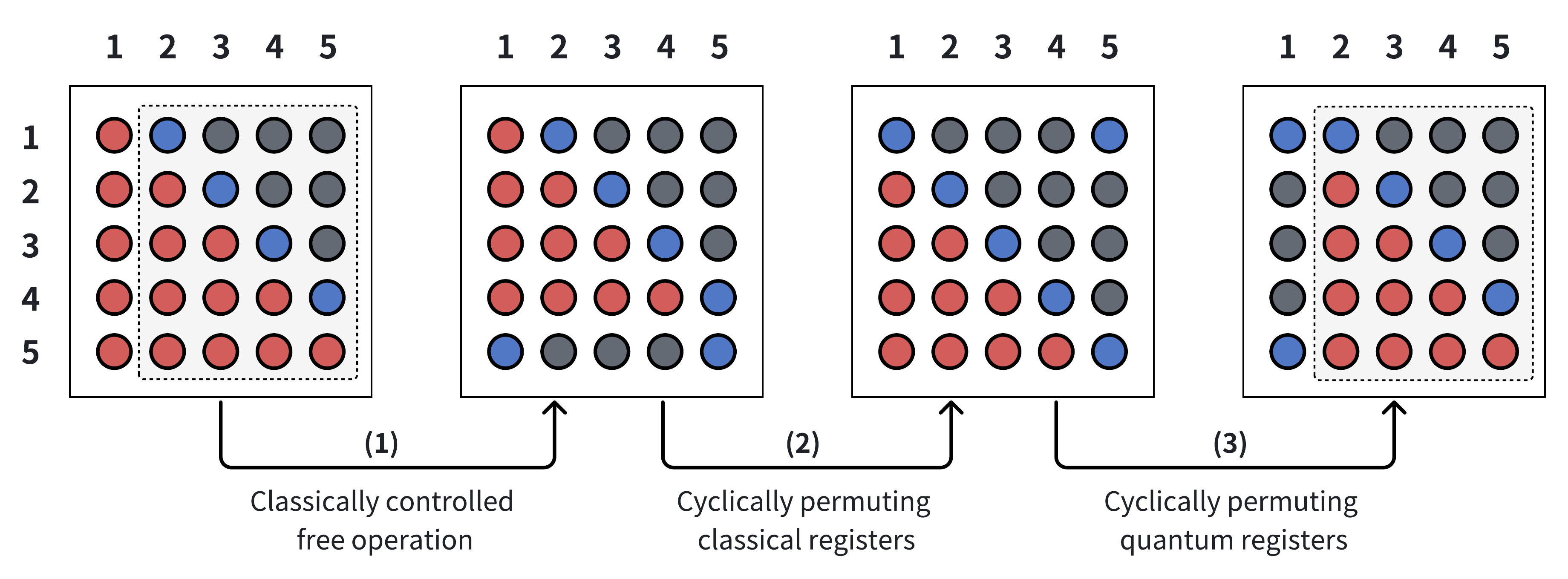}
    \caption{An illustration for the steps of {our catalytic} procedure used to prove Theorem~\ref{thm: reduce the overhead}, for the case where $n=5$, $m=2$, and $k = 1$. Each dot represents a quantum state on the system $T$. Red dots correspond to the embedded source state $\tilde \rho \in \density(T)$, while blue dots represent the embedded target state $\tilde \sigma \in \density(T)$. Grey dots represent the embedded free state $\tilde \theta \in \density(T)$, which is orthogonal to $\tilde \sigma$. The overall quantum state is a mixture of rows, each labeled with a classical register. {Suppose we have a 5-to-2 distillation protocol with success probability $p$.} In the first step, a classically controlled free operation is applied, transforming the last row from $\tilde \rho^{\ox 5}$ to $\tilde \sigma^{\ox 2} \ox \tilde \theta^{\ox 3}$ {with success probability $p$}. The second step involves cyclically permuting the classical registers (i.e., the rows of dots), and the third step involves cyclically permuting the quantum registers (i.e., the columns of dots). {Since $\tilde \sigma$ is orthogonal to $\tilde \theta$, we can perform projective measurement and post-select $\tilde \sigma$ with success probabilty $2/5$. Overall, this gives a 1-to-1 catalytic distillation protocol from $\rho$ to $\sigma$ with success probability $2p/5$.} The dashed boxes highlight the catalyst state, showing that it remains unchanged before and after the transformation.}
    \label{fig:spacetime tradeoff}
\end{figure}

Some remarks concerning our above result are in order. First, {the catalyst construction extends beyond the standard ones in Refs.~\cite{duan2005multiple,char2023catalytic,datta2023there,Kondra2021PRL,Lipka-Bartosik2021,Shiraishi2021,Takagi2022PRL,Wilming2021PRL,Bavaresco2025}.} Appending the free state $\theta$ on the orthogonal subspace is crucial here, as it allows the target state to be post-selected after the projective measurement, ensuring that the fidelity is not compromised in the final step. Second, the catalyst used in the proof is not unique. Other options could also work in the same manner, such as
\begin{align}
\omega = \frac1g \sum_{i=1}^{g-m} \tilde \zeta^{\ox i-1} \ox \tilde \eta_{1:m-1}^{m} \ox \tilde \theta^{\ox g - m - i+1} \ox \ket{i}\bra{i} + \frac1g \sum_{i=g-m+1}^{g} \tilde\zeta^{\ox i-1} \ox \tilde\eta^m_{1:g-i} \ox \ket{i}\bra{i}.
\end{align}
Interestingly, such a catalyst is independent on the target state $\tilde \eta^m$ if $m = 1$. {Third, if $p=1$, the catalyst is always returned on the marginal before the post-selection on the target system. So its status is preserved by ignoring the measurement outcome and therefore can be reused in the subsequent distillation rounds.}

\subsection{Proof of Theorem~\ref{thm: app mutual information}}

The proof of this result requires the following lemmas.

\begin{lemma}\label{lem: DPI of I}
For any quantum channel $\cN$ and superchannel $\Pi$, the output channel $\Pi(\cN)$ satisfies $ I(\Pi(\cN)) \leq I(\cN)$.
\end{lemma}
\begin{proof}
This can be understood from the existing operational interpretation of a channel's mutual information, which represents its entanglement-assisted quantum capacity. Since the channel $\Pi(\cN)$ is noisier than $\cN$, it results in reduced communication capability. A similar argument is provided in Ref.~\cite[Remark 5]{fang2019quantum}.
\end{proof}

\vspace{0.2cm}
Moreover, the quantum channel mutual information is continuous with respect to the diamond norm.

\begin{lemma}\label{lem: continuity of I}
Let $\cN, \cM \in \CPTP(A\to B)$ and $d_{AB}$ be the dimension of $A\ox B$. If $\|\cN - \cM\|_\diamond \leq \ve$, it holds that
\begin{align}
    |I(\cN) - I(\cM)| \leq 3\ve \log d_{AB} + 3 h_2(\ve):= f(\ve).
\end{align}
\end{lemma}
\begin{proof}
For any $\delta > 0$, choose $\varphi_1$ and $\varphi_2$ be such that $|I(\cN) - I(B:R)_{\cN(\varphi_1)}| \leq \delta$ and $|I(\cM) - I(B:R)_{\cN(\varphi_2)}| \leq \delta$ where $d_R = d_A$. Since $\|\cN - \cM\|_\diamond \leq \ve$, we have $\|\cN(\varphi_1) - \cM(\varphi_1)\|_1 \leq \ve$ and $\|\cN(\varphi_2) - \cM(\varphi_2)\|_1 \leq \ve$. Recall the Fannes inequality $|S(\rho)-S(\psi)| \leq \ve \log d + h_2(\ve)$ if $\|\rho-\psi\|_1 \leq \ve$. We have the continuity of the quantum mutual information as 
\begin{align}
|I(B:R)_\rho& - I(B:R)_\psi|\notag\\
 & \leq |S(\rho_B) -S(\psi_B)| + | S(\rho_R) - S(\psi_{R})| +  |S(\rho_{BR}) - S(\psi_{BR})|\\
& \leq 3\ve \log d_{AB} + 3 h_2(\ve):=f(\ve).
\end{align}
if $\|\rho_{BR} - \psi_{BR}\|_1 \leq \ve$.
Putting everything together, we have
\begin{align}
& I(\cN) - \delta \leq I(B:R)_{\cN(\varphi_1)} \leq I(B:R)_{\cM(\varphi_1)} + f(\ve) \leq I(\cM) + f(\ve),\\
& I(\cM) - \delta \leq I(B:R)_{\cM(\varphi_2)} \leq I(B:R)_{\cN(\varphi_2)} + f(\ve) \leq I(\cN) + f(\ve).
\end{align}
So we get $|I(\cN) - I(\cM)| \leq \delta + f(\ve)$. Since $\delta$ can be made arbitrarily small, we get $|I(\cN) - I(\cM)| \leq  f(\ve)$.
\end{proof}

\begin{lemma}\label{lem: marginal channel monotone}
    Let $\cP \in \CPTP(A'X\to B'Y)$, $\cM \in \CPTP(A'\to B')$, and $\cC \in \CPTP(X\to Y)$. Let $\cP_{A'\to B'} (\cdot) := \tr_Y \cP_{A'X\to B'Y}((\cdot) \ox \pi_X)$ be the reduced channel of $\cP$ on the system $A'$. If
    $\Delta(\cP_{A'X\to B'Y}, \cM_{A'\to B'} \ox \cC_{X\to Y}) \leq \ve$, then $\Delta(\cP_{A'\to B'}, \cM_{A'\to B'}) \leq \ve$.
\end{lemma}
\begin{proof}
    
This can be checked as follows:
\begin{align}
 \|\cP_{A'\to B'} & - \cM_{A'\to B'}\|_\diamond\notag \\
 & = \sup_{\rho_{RA'}} \|\cP_{A'\to B'}(\rho_{RA'}) - \cM_{A'\to B'}(\rho_{RA'})\|_1\\
& = \sup_{\rho_{RA'}} \|\tr_Y \cP_{A'X\to B'Y}(\rho_{RA'} \ox \pi_X) - \cM_{A'\to B'}(\rho_{RA'})\|_1\\
 & = \sup_{\rho_{RA'}} \|\tr_Y \left\{\cP_{A'X\to B'Y}(\rho_{RA'} \ox \pi_X) - \cM_{A'\to B'}(\rho_{RA'}) \ox \cC_{X\to Y}(\pi_X)\right\}\|_1\\
 & \leq \sup_{\rho_{RA'}} \|\cP_{A'X\to B'Y}(\rho_{RA'} \ox \pi_X) - \cM_{A'\to B'}(\rho_{RA'}) \ox \cC_{X\to Y}(\pi_X)\|_1\\
 & \leq \sup_{\rho_{RA'X}} \|\cP_{A'X\to B'Y}(\rho_{RA'X}) - \cM_{A'\to B'}\ox \cC_{X\to Y} (\rho_{RA'X})\|_1\\
 & = \|\cP_{A'X\to B'Y} - \cM_{A'\to B'}\ox \cC_{X\to Y}\|_\diamond,\label{eq: diamond norm partial trace}
\end{align}
where the first inequality follows by the data-processing inequality of trace norm and the second inequality follows by relaxing the input state to all quantum states on $RA'X$.
\end{proof}

\vspace{0.2cm}
We are now ready to prove Theorem~\ref{thm: app mutual information}.
We first show the ``only if'' direction. Suppose $\cN$ can be catalytically transformed to $\cM$. For any $\ve > 0$, there exists a catalyst channel $\cC_n$ and a code $\Pi_n$ such that $\Delta(\Pi_n (\cN \ox \cC_n),\cM \ox \cC_n) \leq \ve$. By the continuity of the channel mutual information and the monotonicity in Lemma~\ref{lem: continuity of I} and Lemma~\ref{lem: DPI of I}, we get
\begin{align}
I(\cM \ox \cC_n) & \leq I(\Pi_n (\cN \ox \cC_n)) + f(2\ve) \leq I(\cN \ox \cC_n) + f(2\ve).
\end{align}
Then by the additivity of channel mutual information, we get $I(\cM) \leq I(\cN) + f(2\ve)$. Since the inequality holds for arbitrary $\ve$, we get $I(\cM) \leq I(\cN)$.  

Then we prove the ``if'' direction. That is, if $I(\cN) \geq I(\cM)$ we can find a sequence of $\Pi_n$ and $\cC_n$ that catalytically transform $\cN$ to $\cM$. By the reversibility of the channel simulation, we know that for any given $\ve$, there exists entanglement-assisted or NS-assisted transformations such that 
\begin{align}
    \Pi_n(\cN^{\ox n}) = \cP^n \quad \text{and} \quad \|\cP^n - \cM^{\ox n}\|_\diamond \leq \ve.
\end{align}
Note that $\cP^n \in \CPTP(A^n \to B^n)$. Define the reduced channel on the first $i$ systems as
\begin{align}
    \cP^n_{1:i} (\cdot) := \tr_{B_{i+1:n}} \cP^n_{A^n \to B^n}((\cdot) \ox \pi_{A_{i+1:n}}),
\end{align}
where $\cP^n_{1:n} := \cP^n$.
Now consider a catalyst channel 
\begin{align}
    \cC_n := \frac{1}{n} \sum_{k=1}^n \cN^{\ox k-1} \ox \cP^n_{1:n-k} \ox \ket{k}\bra{k}.
\end{align}
Then the overall channel is given by
\begin{align}
    \cN \ox \cC_n = \frac{1}{n} \sum_{k=1}^n \cN^{\ox k} \ox \cP^n_{1:n-k} \ox \ket{k}\bra{k}.
\end{align}
Let $\cQ_0(\cdot):= Q_0 (\cdot) Q_0$ with $Q_0:= \sum_{i=1}^{n-1} \ket{i}\bra{i}$ and $\cQ_1(\cdot):= Q_1 (\cdot) Q_1$ with $Q_1:= \ket{n}\bra{n}$ be two projecting maps on the flag register $F$.
First, performing a classically controlled transformation $\cI \ox \cQ_0 + \Pi_n \ox \cQ_1$, we get
\begin{align}
    \cG_1 = \frac{1}{n} \sum_{k=1}^{n-1} \cN^{\ox k} \ox \cP^n_{1:n-k} \ox \ket{k}\bra{k} + \frac{1}{n} \cP^n \ox \ket{n}\bra{n}.
\end{align}
Next, we relabel the classical registers $i\to i+1$ and $n \to 1$ and get
\begin{align}
    \cG_2 = \frac{1}{n} \sum_{k=1}^{n} \cN^{\ox k-1} \ox \cP^n_{1:n-k+1} \ox \ket{k}\bra{k}.
\end{align}
Then performing a cyclic permutation $\cS_{A^n}$ and $\cS_{B^n}$ on the input and the output, respectively, such that $i\to i+1$ and $n \to 1$, we get
\begin{align}
    \cG_3 = \frac{1}{n} \sum_{k=1}^{n} \cS_{B^n} \circ (\cN^{\ox k-1} \ox \cP^n_{1:n-k+1}) \circ \cS_{A^n} \ox \ket{k}\bra{k}.
\end{align}
Then we claim that $\cG_3$ returns $\cC_n$ on its marginal and $\|\cG_3 - \cM\ox \cC_n\|_\diamond \leq 2\ve$. The first claim is equivalent to check the first $n-1$ reduced channel of $\cG_2$, 
\begin{align}
    \tr_{B_n} \circ \; \cG_2 ((\cdot) \ox \pi_{A_n}) & = \frac{1}{n} \sum_{k=1}^{n} \tr_{B_n}  \cN^{\ox k-1} \ox \cP^n_{1:n-k+1} ((\cdot) \ox \pi_{A_n}) \ox \ket{k}\bra{k} \\
    & = \frac{1}{n} \sum_{k=1}^{n}  \cN^{\ox k-1} \ox \tr_{B_n} \cP^n_{1:n-k+1} ((\cdot) \ox \pi_{A_n}) \ox   \ket{k}\bra{k} \\
    & = \frac{1}{n} \sum_{k=1}^{n}  \cN^{\ox k-1} \ox \cP^n_{1:n-k} (\cdot)\ox \ket{k}\bra{k} \\
    & = \cC_n.
\end{align}
Note that the second last equality holds by the definition of $\cP^n_{i}$ and here $\pi_{A_n}$ and $\tr_{B_n}$ are effectively acting on the last systems of the channel $\cP_{1:n-k+1}^n$, which is also the $n-k+1$ system in the original definition of $\cP^n_{1:n-k+1}$.
Since the diamond norm is invariant under system permutation, the second claim is equivalent to check $\|\cG_2 - \cC_n \ox \cM\|_\diamond \leq 2\ve$.
This can be shown as follows:
\begin{align}
   \|&\cG_2 -  \cC_n \ox \cM\|_\diamond \notag \\
   & = \left\| \frac{1}{n} \sum_{k=1}^{n} \cN^{\ox k-1} \ox \cP^n_{1:n-k+1} \ox \ket{k}\bra{k} - \frac{1}{n} \sum_{k=1}^{n} \cN^{\ox k-1} \ox \cP^n_{1:n-k} \ox \cM \ox \ket{k}\bra{k}\right\|_\diamond\\
   & = \frac{1}{n} \sum_{k=1}^{n} \left\| \cN^{\ox k-1} \ox \cP^n_{1:n-k+1} - \cN^{\ox k-1} \ox \cP^n_{1:n-k} \ox \cM \right\|_\diamond\\
   & = \frac{1}{n} \sum_{k=1}^{n} \left\|\cP^n_{1:n-k+1} - \cP^n_{1:n-k} \ox \cM \right\|_\diamond\\
   & \leq \frac{1}{n} \sum_{k=1}^{n} \|\cP^n_{1:n-k+1} - \cM^{\ox n-k+1}\|_\diamond + \frac{1}{n} \sum_{k=1}^{n} \|\cP^n_{1:n-k} \ox \cM - \cM^{\ox n-k+1}\|_\diamond\\
   & = \frac{1}{n} \sum_{k=1}^{n} \|\cP^n_{1:n-k+1} - \cM^{\ox n-k+1}\|_\diamond + \frac{1}{n} \sum_{k=1}^{n} \|\cP^n_{1:n-k} - \cM^{\ox n-k}\|_\diamond\\
   & \leq 2\ve.
\end{align}
where the first inequality follows by the triangle inequality of diamond norm and the last inequality follows from the assumption of $\|\cP^n - \cM^{\ox n}\|_\diamond \leq \ve$ and the monotonicity of diamond norm by taking reduced channel in Lemma~\ref{lem: marginal channel monotone}. As this holds for arbitrary $\ve$, we can take $\ve$ to be vanishingly small and conclude the proof.

\section*{Data availability}

No datasets were generated or analyzed during the current study.

\bibliographystyle{ieeetr}
\bibliography{main.bib}

\section*{Acknowledgements} 

We thank Hayata Yamasaki for pointing out the deterministic nature of their magic state distillation protocols described in Ref.~\cite{wills2024constantoverheadmagicstatedistillation}. K.F. is supported in part by the National Natural Science Foundation of China (Grant No. 92470113 and 12404569), the Shenzhen Science and Technology Program (Grant No. QNXMB20250701091826036 and JCYJ20240813113519025), the Shenzhen Fundamental Research Program (Grant No. JCYJ20241202124023031), the General R\&D Projects of 1+1+1 CUHK-CUHK(SZ)-GDST Joint Collaboration Fund (Grant No. GRDP2025-022), and the University Development Fund (Grant No. UDF01003565). 
 Z.-W.L.\ is supported in part by NSFC under Grant No.~12475023, Dushi program, and a startup funding from YMSC.

\end{document}